\documentclass[a4paper,10pt,reqno]{amsart}

\usepackage[centertags]{amsmath}
\usepackage{amsfonts}
\usepackage{euscript}
\usepackage{amssymb}
\usepackage{amsthm}
\usepackage{newlfont}
\usepackage{stmaryrd}
\usepackage{mathrsfs}
\usepackage{euscript}

\theoremstyle{plain}
\newtheorem{theorem}{Theorem}[section]

  \newtheorem{lemma}[theorem]{Lemma}
\theoremstyle{definition}
  \newtheorem{definition}{Definition}[section]

\theoremstyle{remark}

\numberwithin{equation}{section}

 \newcounter{smallarabics}
\newenvironment{arabicenumerate}
{\begin{list}{{\normalfont\textrm{\arabic{smallarabics})}}}
  {\usecounter{smallarabics}\setlength{\itemindent}{0cm}
  \setlength{\leftmargin}{5ex}\setlength{\labelwidth}{4ex}
  \setlength{\topsep}{0.75\parsep}\setlength{\partopsep}{0ex}
   \setlength{\itemsep}{0ex}}}
{\end{list}}

 \DeclareMathOperator{\supp}{Supp}

\let\al=\alpha \let\be=\beta  
   
 \let\la=\lambda \let\om=\omega 
\let\si=\sigma

  \let\La=\Lambda

\newcommand{\caB}{{\mathcal B}}

\newcommand{\caG}{{\mathcal G}}
\newcommand{\caH}{{\mathcal H}}
\newcommand{\caI}{{\mathcal I}}
\newcommand{\caJ}{{\mathcal J}}

\newcommand{\caS}{{\mathcal S}}

\newcommand{\bbC}{{\mathbb C}}

\newcommand{\bbN}{{\mathbb N}}

\newcommand{\bbR}{{\mathbb R}}

\newcommand{\bbZ}{{\mathbb Z}}

\newcommand{\opunit}{\text{1}\kern-0.22em\text{l}}

\newcommand{\frH}{{\mathfrak H}}

\newcommand{\frU}{{\mathfrak U}}

\DeclareMathAlphabet{\mathpzc}{OT1}{pzc}{m}{it}

\newcommand{\id}{\textrm{d}}

\DeclareMathOperator{\Tr}{Tr}

\newcounter{smallroman}

\newcommand{\ben}{\begin{arabicenumerate}}
\newcommand{\een}{\end{arabicenumerate}}

\newcommand{\beq}{\begin{equation}}
\newcommand{\eeq}{\end{equation}}

\newcommand{\baq}{\begin{eqnarray}}
\newcommand{\eaq}{\end{eqnarray}}

\renewcommand{\sp}{{\mathrm{sp}}}

\newcommand{\ran}{\mathrm{Ran} }
\renewcommand{\d}{\mathrm{d}}

\newcommand{\norm}{ \|}
\newcommand{\str}{ |}

\begin{document}

\title[]  {\large A note on the non-commutative Laplace--Varadhan integral lemma}

\maketitle

\begin{center}
{\bf Wojciech De Roeck}\footnote{Postdoctoral Fellow FWO, email: {\tt wojciech.deroeck@fys.kuleuven.be}}\\
Institut f\"ur Theoretische Physik, ETH Z\"urich \\
Instituut voor Theoretische Fysica, K.U.Leuven
\\ \vspace{10pt}
{\bf Christian Maes}\footnote{email: {\tt
christian.maes@fys.kuleuven.be}} \\
Instituut voor Theoretische Fysica, K.U.Leuven\\\vspace{10pt}
{\bf Karel Neto\v{c}n\'{y}\footnote{email: {\tt netocny@fzu.cz}}}\\
Institute of Physics AS CR, Prague \\\vspace{10pt}
{\bf Luc Rey-Bellet\footnote{email: {\tt luc@math.umass.edu}}} \\
Department of Mathematics \& Statistics \\
University of Massachusetts, Amherst          \\\vspace{10pt}
\end{center}

\begin{abstract}
We continue the study of the free energy of quantum lattice spin
systems where to the local Hamiltonian $H$ an arbitrary mean field
term is added, a polynomial function of the arithmetic mean of
some local observables $X$ and $Y$ that do not necessarily commute. 
By slightly extending a  recent paper by Hiai, Mosonyi, Ohno and Petz \cite{hiaimosonyiohno}, we prove in  general  that the free energy is given by a variational principle over the range of the operators $X$ and $Y$.
As in \cite{hiaimosonyiohno}, the  result is a noncommutative extension of the Laplace-Varadhan
asymptotic formula.
\end{abstract}

\section{Introduction}\label{sec: introduction}

\subsection{Large deviations}
One of the highlights in the combination of analysis and
probability theory is the asymptotic evaluation of certain
integrals.  We have here in mind  integrals of the form, for some
real-valued function $G$,
\begin{equation}\label{ld}
\int \d \mu_n( x) \exp \{ v_n G(x) \} ,\qquad v_n\nearrow +\infty\,\mbox{
as } n\nearrow +\infty
\end{equation}
for which the measures $\mu_n$ satisfy a law of large numbers.
Such integrals can be evaluated depending on the asymptotics of
the $\mu_n$.  The latter is the subject of the theory of large
deviations, characterizing the rate of convergence in the law of
large numbers.  In a  typical scenario, the $\mu_n$ are the
probabilities of some macroscopic variable, such as the average
magnetization or the particle density in ever growing volumes
$v_n$ and as distributed in a given equilibrium Gibbs ensemble.
Then, depending on the case, thermodynamic potentials $\caJ$ make
the rate function $\d\mu_n( x) \sim  \id x \exp \{ -v_n\caJ(x)  \} $ in
the sense of large deviations for Gibbs measures, see
\cite{ollagibbsrandomfields,ellis,georgiibook,varadhanasymptoticprobabilities,varadhanlargedeviations}. That theory of large deviations
is however broader than the applications in equilibrium
statistical mechanics. Essentially, when the rate function for
$\mu_n$ is given by $\caJ$,  then the integral~\eqref{ld} is computed as
\beq \label{eq: variation from varadhan}  \frac{1}{v_n} \log \int \id \mu_n(x) \exp\{ v_n G(x ) \}    \quad \mathop{\longrightarrow}\limits_{n \nearrow +\infty} \quad   \sup_x \{ G(x)-\caJ(x) \}   \eeq
This is
a typical application of Laplace's asymptotic formula for the
evaluation of real-valued integrals. The systematic combination with the
theory of large deviations gives the so called Laplace-Varadhan
integral lemma.

We first  recall the large deviation principle (LDP). Let $(M,d)$ be some complete separable metric
space.
 \begin{definition}
 The sequence of measures $\mu_n$  on $M$ satisfies a LDP
 with rate function $\caJ: M \to \bbR^+\cup\{ +\infty\}$ and speed $v_n \in \bbR^+$ if
 \ben
 \item{$\caJ$ is convex and has closed level sets, i.e.,
  \beq \label{level sets} \{ \caJ^{-1}(x) , \, x \leq c  \} \eeq is closed in
   $(M,d)$ for all $c \in \bbR^+$;  }
  \item{ for all Borel sets $U \subset M$ with interior $\mathrm{int} \,U$
  and closure $\mathrm{cl}\, {U}$, one has
  \begin{eqnarray}
           \liminf_{n \nearrow +\infty} \frac{1}{v_n}  \log \mu_n(U)   &\geq&    -  \inf_{ u \in \mathrm{int} \,U} \caJ(u)  \nonumber  \\
             \limsup_{n \nearrow +\infty}
              \frac{1}{v_n}\log \mu_n(U)         &\leq&  -  \inf_{ u \in \mathrm{cl} \,{U}} \caJ(u) \nonumber
  \eaq
  }
 \een
We say that the rate function $\caJ$ is \emph{good} whenever the level sets  \eqref{level sets} are compact.
\end{definition}

For the transfer of LDP, one considers a pair $(\mu_n,\nu_n)$, $n \nearrow \infty$ of sequences of absolutely continuous measures on $(M,d)$ such that
\[
\frac{\id \nu_n}{\id \mu_n} (x) = \exp \{v_n\,G(x) \} ,\quad\;
\mu_n-\mbox{almost everywhere}
\]
for some measurable mapping $G: M \rightarrow \bbR$.
We now state an instance of the Laplace-Varadhan lemma.

\begin{lemma}[Laplace-Varadhan integral lemma]\label{lv}
Assume that $G$ is bounded and continuous and that the sequence $(\mu_n)$
satisfies a large deviation principle with good rate function
$\caJ$ and speed $v_n$. Then $(\nu_n)$ satisfies a large deviation
principle with good rate function $G - \caJ$ and speed $v_n$.
\end{lemma}

For  more general versions and proofs we refer to the literature,
see e.g. \cite{varadhanasymptoticprobabilities,varadhanlargedeviations,deuschelstroock,dembozeitouni,denhollander}; it remains an important
subject of analytic probability theory to extend the validity  of
the variational formulation   \eqref{eq: variation from varadhan}  and to deal with its applications.

\subsection{Mean-field interactions}

From the point of view of equilibrium statistical
mechanics, one can also think of the formula~\eqref{ld} as giving (the exponential of)
the pressure or free energy when adding a mean field type term 
to a  Hamiltonian which is a sum of local interactions.

The choice of the function $G$ is
then typically monomial with a power decided by the number of
particles or spins that are in direct interaction.  For example,
the free energy of an Ising-like model with such an extra mean field interaction would be given by the limit
\begin{equation}\label{ising}
\lim_{\La \nearrow \bbZ^d}  \frac 1{|\La|} \log \sum_{\eta \in \{+,-\}^\La} \exp
\Bigl(-\beta H_\La(\eta) + \lambda_p\,|\La|\,
\Bigl(\frac{1}{|\La|} \sum_{i\in \La} \eta_i \Bigr)^p \Bigr)
\end{equation}
for $p=1,2,\ldots$, where $H_\La(\eta)$ is the (local) energy of
the spin configuration $\eta$ and the limit takes a sequence of
regularly expanding volumes $\La$ to cover some given lattice.  The case $p=1$
corresponds to the addition of a magnetic field $\lambda_1$; $p=2$
is most standard and adds effectively a very small but long range
two-spin interaction.  Higher $p-$values are also not uncommon in
the study of Ising interactions on hypergraphs, and even very
large $p$ has been found relevant e.g. in models of spin
glasses and in information theory \cite{bolleheylenskantzos}.

The form \eqref{ld} is easily recognized in \eqref{ising}, with
\[ \mu_n (x) \sim  \sum_{ \eta\in \{+,-\}^\La,\, \sum_{i \in \La} \eta_i =x \str\La\str }\exp \{ -\beta H_\La(\eta )\}, \qquad   v_n = |\La| \] and
the function $G(x)= \lambda_p\,x^p$.  The Laplace-Varadhan lemma  applies
to~\eqref{ising} since we know that the sequence of Gibbs states
with density $\sim \exp \{-\beta H_\La(\cdot) \}$ satisfies a LDP
with a good rate function $\caJ_\text{cl}$  and speed $\str \La
\str$. The result reads that \eqref{ising} is given by the variational formula
\beq
\sup_{u \in [-1,1]} \{ \la_p u^p - \caJ_\text{cl}(u)   \}
\eeq


In noncommutative versions the local Hamiltonian $H$ and the additional mean field term are allowed not to commute with each other. That is natural within the statistical mechanics of quantum spin systems and this is also the context of the present paper.

\subsection{Noncommutative extensions} \label{sec: noncomm extensions}

Although it has proven very useful to think of integrals~\eqref{ld} within the framework of  probability and large deviation theory, it is fundamentally a problem of analysis.
However, without such a probabilistic context, the question of a noncommutative extension of the Laplace-Varadhan Lemma~\ref{lv} becomes ambiguous and it in fact allows for different formulations, each possibly having a physical interpretation on its own.

One approach is to ask for the asymptotic evaluation of the expectations
\begin{equation}\label{var}
\lim_{\La \nearrow \bbZ^d} \frac{1}{|\La|} \log
  \om_\La\bigl( e^{|\La|\, G(\bar X_\La)} \bigr)
  \end{equation}
under a family of quantum states $\omega_\La$ where
$\bar X_\La$ would now be the arithmetic mean
of some quantum observable in volume $\La$.
To be specific, one can take
$\om_\La\sim \exp \{ -\beta H_\La \}$ a quantum Gibbs state for a quantum
Hamiltonian $H_\La$ and $\bar X_\La = (\sum_{i\in
\La}X_i)/|\La|$ the mean magnetization in some fixed direction. Arguably, this formulation is closely related to the asymptotic statistics of outcomes in von Neumann measurements of $\bar X_\La$, \cite{netocnyredig,reybelletld,ogata,hiaimosonytomohiro}.

A more general class of possible extensions is obtained by considering the limits of
\begin{equation}\label{q2}
  \frac{1}{|\La|} \log  \Tr_\La \bigl(
  \si_\La^{\frac{1}{K}}\,e^{\frac{|\La|}{K}\, G(\bar X_{\La})} \bigr)^K\,,\qquad
  \La \nearrow \bbZ^d
\end{equation}
for different $K > 0$, where $\si_\La$ is the density matrix of a quantum state in volume $\La$. For the canonical form
$\si_\La = \exp(-\beta H_\La) / Z^\beta_\La$ with
local Hamiltonian $H_\La$ at inverse temperature $\beta$, \eqref{q2}
becomes
\begin{equation}
  \frac{1}{|\La|} \log \frac{1}{Z^\be_{\La}} \Tr_\La \bigl(
  e^{-\frac{\be}{K} H_\La}\,e^{\frac{|\La|}{K}\, G(\bar X_\La)} \bigr)^K\,,\qquad
  \La \nearrow \bbZ^d
\end{equation}
There is no \emph{a priori} reason to exclude any particular value of $K$ from consideration.  Two standard options are:    $K = 1$, which corresponds to the expression~\eqref{var} above, and $K \nearrow +\infty$, which, by the Trotter product formula, boils down to
\begin{equation}\label{eq: K-infty}
  \frac{1}{|\La|} \log \frac{1}{Z^\be_{\La}} \Tr_\La
  \bigl( e^{-\be H_\La + |\La|\,G(\bar X_\La)} \bigr)\,,\qquad
  \La \nearrow \bbZ^d
\end{equation}
which is the free energy of a corresponding quantum spin model, cf.~\eqref{ising}.
In the present paper, we study the case $K \nearrow +\infty$  (without touching the question of   interchangeability of both limits).

One of our results, Theorem~\ref{thm: main} with $Y=\bar Y_\La=0$, is  of the form
\begin{equation}\label{res3}
  \lim_{\La \nearrow \bbZ^d} \frac{1}{|\La|} \log \Tr_\La
  \bigl( e^{-\be H_\La + |\La|\,G(\bar X_{\La})} \bigr)
  = \sup_{-\|X\| \leq u \leq \|X\|} \{G(u) - \cal J(u)\}
\end{equation}
Note that we omitted the normalization factor $ 1/ Z^\be_{\La}$ since it merely adds a constant (independent of $G$) to \eqref{eq: K-infty}.
In the usual context of the theory of large deviations  formula \eqref{res3}
arises as a change of rate function.  However, while our  result
\eqref{res3} very much looks like Varadhan's formula in Lemma
\ref{lv}, there is a big difference in interpretation: The
function $\cal J$ is not as such the rate function of large
deviations for $\bar X_{\La}$.  Instead, it is given as the Legendre transform 
\beq
\caJ(u)=\sup_{t \in \bbR}\, \{  t u - q(t)\} , \qquad   u  \in \bbR
\eeq
of a function $q(\cdot) $  which is the pressure corresponding to a linearized interaction, i.e.\ 
\beq\label{def: pressure}
  q(t) =   \lim_{\La \nearrow \bbZ^d} \frac{1}{|\La|}  \log   \Tr_\La
  \bigl( e^{-\be H_\La +  t |\La| \bar X_{\La})} \bigr)\eeq 

\subsection{Several non-commuting observables: Towards joint large deviations?} \label{sec: several noncomm}

In the previous Section \ref{sec: noncomm extensions}, we made the tacit assumption that there is a single observable $\bar X_\La$ corresponding to some operator on Hilbert space. However, in formula \ref{ising},  the observable $\frac{1}{|\La|} \sum_{i\in \La} \eta_i $ could equally well represent a vector-valued magnetization which, upon quantization, would correspond to several non-commuting observables $\bar X_{\La}, \bar Y_{\La}$, say, the magnetization along the $x$ and $y$-axis, respectively.  In the commutative theory, this case does not require special attention; the framework of large deviations applies equally regardless of whether the observable takes values in  $\bbR$ or $\bbR^2$.   
Obviously, this is not true in the noncommutative setting and in fact, we do not even know a natural analogue of the generating function \eqref{var}, since we do not dispose of a simultaneous Von Neumann measurement of $\bar X_\La$ and $\bar Y_\La$.   One can take the point of view that this is inevitable in quantum mechanics, and insisting is pointless.
Yet,  as $\La \nearrow \bbZ^d$, the commutator  
\beq
 [\bar X_\La, \bar Y_\La ] = O(\frac{1}{\str \La \str})
\eeq
vanishes and hence the joint measurability of $\bar X_\La, \bar Y_\La $ is restored on the macroscopic scale.  We refer the reader to \cite{deroeckmaesnetocnyqhthm} where this issue is discussed and studied in more depth. 

The advantage of the approach via the  Laplace-Varadhan Lemma is that one can set aside these conceptual questions and study 'joint large deviations' of $\bar X_\La$ and $\bar Y_\La$ by 
 choosing $G$ to be a joint function of $\bar X_\La$ and $\bar Y_\La$, for example a symmetrized monomial 
 \beq \label{eq: monomial}
 G(\bar X_\La, \bar Y_\La) =  (\bar X_\La)^k (\bar Y_\La)^l +  (\bar X_\La)^l (\bar Y_\La)^k, \qquad  \textrm{for some} \, \,  k,l  \in \bbN
 \eeq
 and check whether the formula \eqref{res3} remains valid with some obvious adjustments. 
This turns out to be the case and it is our main result: Theorem \ref{thm: main}.

\subsection{Comparison with previous results}

The asymptotics of the expression~\eqref{eq: K-infty} was first
studied and the result \eqref{res3} was first obtained by Petz \emph{et al.\ }\cite{petzraggioverbeure}, in the case where
the Hamiltonian $H_\La$ is made solely from a one-body
interaction.  The corresponding equilibrium state is then a product
state. In~\cite{hiaimosonyiohno}, Hiai \emph{et al.\ }  generalized this result to the case of
locally interacting spins but the lattice dimension was restricted
to $d=1$.  However, the authors of \cite{hiaimosonyiohno} argue
that the restriction to $d=1$ can be lifted in the
high-temperature regime. The main reason is that their work relies
heavily on an asymptotic decoupling condition which is proven in that regime, \cite{arakiion}.  One should observe here that this asymptotic decoupling condition
 in fact  implies a large deviation principle for $\bar X_{\La}$, as
 follows from the work of Pfister \cite{pfisterlattice}. Hence, in the language of Section \ref{sec: noncomm extensions},  \cite{hiaimosonyiohno}  evaluates \eqref{eq: K-infty} (the case $K=\infty$) in those regimes where the \eqref{var} (the case $K=1$) is analytic.  

The present paper elaborates on the result of \cite{hiaimosonyiohno}  in two ways. First, we remark that, in our setup, the decoupling condition is actually not necessary for \eqref{res3} to hold, and therefore one can do away with the restriction to $d=1$ or high temperature.  Hence, again referring to Section \ref{sec: noncomm extensions}, the case $K=\infty$ can be controlled even when we know little about the case $K=1$.  To drop the decoupling condition, it is absolutely essential that we start from finite-volume Gibbs states, and not from finite-volume restrictions of infinite-volume Gibbs states, as is done in \cite{hiaimosonyiohno}. 

Second, we show that by the same formalism, one can treat the case of several noncommuting observables, as explained in Section \ref{sec: several noncomm}.  The most  serious step in this generalization is actually an extension of the result of \cite{petzraggioverbeure} to noncommuting observables. This extension is stated in Lemma \ref{lem: PRV} and proven in Section \ref{sec: proof of PRV}.  \\

\noindent \textbf{Note} 
While we were finishing this paper, we learnt of a  similar project by W. de Siqueira Pedra and J-B. Bru. Their result \cite{brupedra} is nothing less than a full-fledged theory of equilibrium states with mean-field terms in the Hamiltonian, describing not only the mean-field free energy  (as we do here), but also the states themselves.   Also, their results hold for fermions, while ours are restricted to spin systems, and they provide interesting examples.    Yet, the focus of our paper differs from theirs and our main result is not contained in their paper.

%

\subsection{Outline}
In Section \ref{sec: setup}, we sketch the setup. We introduce spin systems on the lattice, noncommutative polynomials and ergodic states.  
 Section \ref{sec: results} describes the result of the paper. 
 The remaining Sections \ref{sec: approx by ergodic},  \ref{sec: lower bound}, \ref{sec: upper bound} and \ref{sec: proof of PRV} contain the proofs.

\section{Setup} \label{sec: setup}

\subsection{Hamiltonian and observables} \label{sec: hamiltonian and observables}
We consider a quantum spin system on the regular lattice
$\bbZ^d$, $d = 1,2,\ldots$. We briefly introduce the essential setup below, and we refer to~\cite{israel,simon} for more expanded, standard introductions.

The single site Hilbert space $\cal H$ is finite-dimensional (isomorphic to $\bbC^n$) and for any finite volume $\La \subset \bbZ^d$, we set
$\cal H_\La = \otimes_{ \La} {\cal H}$.  
The $C^*$-algebra of bounded operators on $\caH_\La$ is denoted by
$\caB_\La \equiv \caB(\caH_\La)$. The standard embedding
$\caB_\La \subset \caB_{\La'}$ for $\La \subset \La'$ is assumed throughout.
The quasi-local algebra $\frU$ is defined as the norm closure  of the finite-volume algebras 
\beq
\frU := \overline{\mathop{\bigcup}\limits_{\La \, \, \mathrm{finite} } \caB_{\La}  }
\eeq

Denote by $\tau_i$, $i \in \bbZ^d$, the translation which
shifts all observables over a lattice vector     $i$, i.e.,  $\tau_i$ is a homomorphism from $\caB_\La$ onto $\caB_{i+\La}$.

We introduce an interaction potential $\Phi$, that is a
collection $(\Phi_A)$ of Hermitian elements of $\cal B_A$, labeled
by finite subsets $A \subset \bbZ^d$.  We assume translation
invariance (i) and a finite range (ii):
\ben
\item[i)] $ \tau_i(\Phi_A)=\Phi_{i+A} $  for all finite $A \subset \bbZ^d$;  \vspace{1mm}
\item[ii)]   there is a $d_{max} < \infty$ such that, if $\mathrm{diam}(A) > d_{max}$, then $\Phi_A=0$. 
\een
\vspace{1mm}
In estimates, we will frequently use the number 
\beq
r(\Phi):=  \sum_{ A \ni 0}    \norm \Phi_A \norm  <\infty
\eeq
The local Hamiltonian in a finite volume $\La$ is
\begin{equation}\label{intera}
H_{\La} \equiv H_\La^\Phi = \sum_{A\subset \La} \Phi_A
\end{equation}
which corresponds to free or open boundary conditions. Boundary
conditions will however turn out to be irrelevant for our
results. We will drop the superscript $\Phi$ since we will keep the interaction potential fixed.

 Let $X,Y, \ldots$ denote local observables on the lattice, located at the origin, i.e., $\supp X$ (which is defined as the smallest set $A$ such that $X \in \caB_A$) is a finite set which includes $0 \in \bbZ^d$.

We write
\beq
X_\La :=  \mathop{\sum}\limits_{j \in \bbZ^d,  \supp \tau_j X \subset \La}    \tau_j X
\eeq
and
\beq
\bar X_\La :=  \frac{1}{\str \La \str} X_\La
\eeq
for the corresponding intensive observable (the `empirical average' of $X$). 

All of these operators are naturally embedded into the quasi-local algebra $\frU$.  At some point, we will also require the infinite volume intensive observable
\beq
\bar X  \sim   \bar X_{\La \nearrow \infty} \nonumber
\eeq
Some care is required in dealing with $\bar X$ since it does not belong to the  quasi-local algebra $\frU$.  We will further comment on this in Section \ref{sec: states}.

\subsection{Noncommutative polynomials}\label{sec: noncommutative functions}

We will perturb the Hamiltonian $H^{\Phi}_{\La}$ by a mean field term of the form $\str \La \str G(\bar X_{\La},\bar Y_{\La})$ where $G$ is a  `noncommutative polynomial' of the operators $\bar X_{\La},\bar Y_{\La}$, e.g.\  as in \eqref{eq: monomial}.

In this section, we introduce  these noncommutative polynomials  $G$ as quantizations of polynomial functions $g$.
First, we define 
\beq
\ran (X, Y):= [-\norm X \norm, \norm X \norm] \times [-\norm Y \norm, \norm Y \norm]
\eeq
This definition is motivated by the fact that  (`$\sp$' stands for spectrum)
\beq
\sp \bar X_{\La} \times \sp \bar Y_{\La}    \subset  \ran (X, Y), \qquad  \textrm{for all} \,\, \La
\eeq

Let $g$ be a  real polynomial function on the rectangular  set $\ran (X,Y)$. Using the symbol $\caI$ for the collection of all finite sequences from the binary set $\{1,2\}$, any map
$\tilde G: \caI \longrightarrow \bbC$ is called a quantization of $g$ whenever
\begin{equation}  \label{cond: g}
  \sum_{n = 0}^{N}  
  \sum_{\al = (\al(1),\ldots,\al(n)) \in \caI}
  \tilde G(\al)\,x_{\al(1)} \ldots x_{\al(n)} = g(x_1,x_2)
\end{equation}
for all $(x_1, x_2) \in \ran(X,Y)$ and for some $N \in \bbN$.
A quantization $\tilde G$ is called symmetric whenever 
\beq
 \overline{\tilde G (\al(1),\ldots,\al(n)) }=   \tilde G(\al(n),\ldots,\al(1)).
\eeq
Any such symmetric quantization $\tilde G$ defines a self-adjoint operator
\begin{equation}
  G(X,Y) = \sum_{n = 0}^{N} 
  \sum_{\al = (\al(1),\ldots,\al(n)) \in \caI}
  \tilde G(\al)\,X_{\al(1)} \ldots X_{\al(n)} 
\end{equation}
taking $X_1 \equiv X$ and $X_2 \equiv Y$.

In the thermodynamic limit, one expects different quantizations of $g$ to be equivalent:
\begin{lemma}
Let $\tilde G$ and $\tilde G'$ be any two  quantizations of 
$g: \ran(X,Y) \longrightarrow \bbR$. Then
\beq
  \|G(\bar X_\La, \bar Y_\La) - G'(\bar X_\La, \bar Y_\La) \| \leq 
  \frac{C(X,Y)}{\str \La \str}
\eeq
for some $C_{}(X,Y) < \infty$, and for all finite volumes $\La$.
\end{lemma}
\begin{proof}
This is a simple consequence of  the fact that the commutator of macroscopic observables vanishes in the thermodynamic limit, more precisely, 
\beq
 \norm [\bar X_\La, \bar Y_\La] \norm   \leq  \frac{1}{\str \La \str}  \norm X \norm \str \supp X \str  \times  \norm Y \norm \str \supp Y \str. 
\eeq
\end{proof}

Indeed, our results, Theorems \ref{thm: main} and \ref{prop: mean field var principle}, do not depend on the choice of quantization. This can also be checked a priori using the above lemma and the log-trace inequality in \eqref{eq: log-trace inequality}.

\subsection{Infinite-volume states} \label{sec: states}

A state $\om_\La$ is a positive linear functional on $\caB_\La$, normalized by $\norm \om_\La \norm = \om_\La(1)=1$.
An example is the tracial state, $\om_\La (\cdot) \sim \Tr_\La (\cdot)$.  In general we consider states $\om_\La$ as characterized by their density
matrix $\sigma_\La$, $\om_\La(\cdot) = \Tr_\La (\sigma_\La\cdot)$. 

An infinite volume state $\om$ is a a positive normalized function on the $C^*$-algebra $\frU$ (the quasi-local algebra). It is translation invariant when $\om(A)=\om(\tau_j A)$ for all $j \in \bbZ^d$ and $A \in \frU$.
A translation-invariant state $\om$  is \emph{ergodic} whenever it is an extremal point in the convex set of translation invariant states. 
A state is called \emph{symmetric}  whenever it is invariant under a permutation of the lattice sites, that is, for any sequence of one-site observable $A_1, \ldots, A_n  \in \caB_{\{0\}} \subset \frU $ and $i_1, \ldots, i_n \in \bbZ^d$ 
\beq
\om(\tau_{i_1}( A_1)  \tau_{i_2}( A_2)  \ldots   \tau_{i_n}( A_n) )  =  \om(\tau_{i_{\pi(1)}}( A_1)  \tau_{i_{\pi(2)}}( A_2)  \ldots   \tau_{i_{\pi(n)}}( A_n) )  \eeq
for any permutation $\pi$ of the set $\{1,\ldots, n \}$. The set of ergodic, resp.\ symmetric states on $\frU$ is denoted by $\caS_{\mathrm{erg}}, \caS_{\mathrm{sym}} $, respectively. 

At some point we will need the theorem by St{\o}rmer \cite{stormer} that states that any $\om \in \caS_{\mathrm{sym}}$ can be decomposed as
\beq \label{eq: stormer}
\om = \int_{\textrm{prod.}} \d \nu_{\om}(\phi) \phi
\eeq
for some probability measure $\nu_\om$ on the set of product states.  Of course, the set of product states can be identified with the (finite-dimensional) set of states on the one-site algebra $\caB_{\{0\}}=\caB(\caH) $.

For a finite-volume state $\om_\La$ on $\caB_{\La}$, we consider  the entropy functional
\beq S(\om_\La)\equiv S_{\La}(\om_\La) = - \Tr \si_\La \log \si_\La   \eeq 
  The mean entropy of translation-invariant infinite-volume state $\om$ is denoted as 
  \beq
  s(\om) := \lim_{\La \nearrow \bbZ^d}   \frac{1}{\str \La \str}   S(\om_\La), \qquad \textrm{with}\, \, \om_{\La} := \om\big\str_{\caB_\La} \, (\textrm{restriction to}\, \La)
  \eeq
  In this formula and in the rest of the paper, the limit $\lim_{\La \nearrow \bbZ^d}$ is meant in the sense of  \emph{Van Hove}, see e.g.\ \cite{israel, simon}.
  Standard properties of the functional $s$ are its affinity and upper semicontinuity (with respect to the weak topology on states).

In Section \ref{sec: hamiltonian and observables}, we mentioned the 'observables at infinity' $\bar X$ and $\bar Y$, postponing their definition to the present section. 
Expressions like $
\om(\bar X^{l} \bar Y^{k}) $  (for some positive numbers $l,k$)
can be defined as
\beq \label{eq: limit k and l}
\om(\bar X^{l}\bar Y^{k})  :=  \lim_{\La, \La' \nearrow \bbZ^d}    \om(\bar X^{l}_{\La}  \bar Y^{k}_{\La'}),
\eeq
provided that the limit exists.  We use the following standard result that can be viewed as a noncommutative law of large numbers
\begin{lemma} \label{lem: concentration}
For $\om \in \caS_{\mathrm{erg}}$, the limit  \eqref{eq: limit k and l} exists and 
\beq \label{eq: concentration}
\om(\bar X^{l} \bar Y^{k}) = [\om(X)]^l   [\om(Y)]^k 
\eeq
\end{lemma}
Note that $\om(X)=\om(\bar X)$ and $\om(Y)=\om(\bar Y)$ by translation invariance.
An immediate corollary is that for  a noncommutative polynomial $G$ which is  a quantization of $g$ (see Section \ref{sec: noncommutative functions}), and $\om \in \caS_{\mathrm{erg}}$:
\beq \label{eq: concentration general}
\om( G(\bar X, \bar Y)) = g(\om(X), \om(Y))
\eeq

For the convenience of the reader, we sketch the proof of Lemma \ref{lem: concentration} in Appendix  \ref{app: concentration}

Finally, we note that Lemma \ref{lem: concentration} does not require the state $\om$ to be  trivial at infinity.  
Triviality at infinity is a stronger notion which is not used in the present paper.  In particular, the state $\bar{\mu}$ constructed in Section \ref{sec: approx by ergodic} is ergodic, but not trivial at infinity, since it fails to be ergodic with respect to  a subgroup of lattice translations.

\section{Result} \label{sec: results}
Choose  $X,Y$ to be  local operators and let $H^{\Phi}_\La$ be the Hamiltonian corresponding to a finite-range, translation invariant interaction $\Phi$, as in Section \ref{sec: hamiltonian and observables}.
Let $G(\cdot, \cdot)$ be a symmetric quantization of a polynomial $g$ on the rectangle $\ran (X, Y)$, as defined in Section \ref{sec: noncommutative functions},  and define the ``$G$-mean field partition function"
\beq \label{def: ZG}
Z^G_{\La}(\Phi) :=  \Tr_\La
  \bigl( e^{-H_\La+ |\La|\,G(\bar X_\La,\bar Y_\La )} \bigr)
\eeq
with $\bar X_\La, \bar Y_\La$ empirical averages of $X,Y$.
The following theorem is our main result:
%

\begin{theorem}\label{thm: main}
Define the pressure
\begin{equation} \label{def: pressure2}
  p(u,v) = \lim_{\La \nearrow \bbZ^d}
  \frac{1}{|\La|} \log \Tr_\La e^{-H_\La^\Phi + u X_\La + v Y_\La}
\end{equation}
and its Legendre transform
\begin{equation}
  I(x,y) = \sup_{(u,v) \in \bbR^2} (u x + v y - p(u,v))
\end{equation} 
Then 
\begin{equation}\label{eq: LD-numbers}
  \lim_{\La \nearrow \bbZ^d} 
  \frac{1}{\str \La \str} \log  Z^G_{\La}(\Phi) =
  \sup_{(x,y) \in \bbR^2} (g(x,y) - I(x,y))
\end{equation}
where  the limit $\La \nearrow \bbZ^d$ is in the sense of Van Hove, as in \eqref{def: pressure2}. In particular, the LHS of
\eqref{eq: LD-numbers}  does not depend on the particular form of quantization taken.
\end{theorem}

As discussed in  Section \ref{sec: introduction}, our result expresses the pressure of the mean field Hamiltonian through a variational principle. 
To  derive this result, it is helpful to represent this pressure first as a variational problem on a larger space, namely that of ergodic states, as in Theorem \ref{prop: mean field var principle}. Theorem \ref{thm: main} follows then by parametrizing these states by their values on $X $ and $ Y$.

We also need the 'local energy operator' associated to the interaction $\Phi$ as
\beq \label{def: energy of interaction}
E_{\Phi} :=  \sum_{A\ni 0} \frac{1}{\str A \str} \Phi_A. 
\eeq

\begin{theorem}[Mean-field variational principle]\label{prop: mean field var principle}
Let $s(\cdot)$ be the mean entropy functional, as in Section \ref{sec: states}. Then
\beq \label{eq: prop mean field variational principle}
\lim_{\La \nearrow \bbZ^d} \frac{1}{\str \La \str} \log  Z^G_{\La}(\Phi) = \mathop{\sup}_{\om \in  \caS_{\mathrm{erg}} } \left( g( \om(X),\om(Y)) +  s(\om)- \om(E_{\Phi})      \right)
\eeq
\end{theorem}

To understand how the first term on the RHS of \eqref{eq: prop mean field variational principle} originates from \eqref{def: ZG}, we recall  the equality \eqref{eq: concentration general}   for ergodic states $\om$.

The proof of Theorem \ref{prop: mean field var principle}  is postponed to Sections  \ref{sec: lower bound} and \ref{sec: upper bound}.  Here we prove that Theorem \ref{thm: main} is a rather immediate  
consequence of Theorem \ref{prop: mean field var principle}. 

\begin{proof}[Proof of Theorem \ref{thm: main}]
Writing the right-hand side of~\eqref{eq: prop mean field variational principle} in the form
\begin{equation}
  \sup_{(x,y) \in \bbR^2} (g(x,y) - \tilde I(x,y))
\end{equation}
where
\begin{equation}
  \tilde I(x,y) = \inf_{\om \in  \caS_{\mathrm{erg}}    \atop
  \om(X) = x,\,\om(Y) = y}
  (-s(\om) + \om(E_\Phi))
\end{equation}
is a convex function on $\bbR^2$, infinite on the complement of $\ran(X,Y)$.  It is lower semi-continuous by the lower semi-continuity of $-s$. By using the infinite-volume Gibbs variational principle~\cite{simon, israel}, its Legendre-Fenchel transform reads
\baq
  \sup_{(x,y) \in \bbR^2}  (u x + v y - \tilde I(x,y))
  &=&\mathop{\sup}_{\om \in  \caS_{\mathrm{erg}} }
  (s(\om) - \om(E_\Phi) + u\, \om(X) + v\, \om(Y)) \nonumber
\\
  &=&  p(u,v)
\eaq
The equality $I=\tilde I$ then follows by the involution property of the Legendre-Fenchel transform on the set of convex lower-semicontinuous functions, see e.g.~\cite{simon}.  
\end{proof}

\textbf{Independence of boundary conditions.} Observe that both
Theorem \ref{thm: main} and Theorem \ref{prop: mean field var principle} have
been formulated for the finite volume Gibbs states with open
boundary conditions. It is however easy to check that this choice
is not essential and other equivalent formulations can be
obtained. Indeed, by the standard log-trace inequality, \beq
\label{eq: log-trace inequality}
   \Bigl| \log \Tr_\La
  \bigl( e^{-\be H_\La  + W_\La + |\La|\,G(\bar X_\La, \bar Y_\La)} \bigr)  -   \log  \Tr_{\La}
  \bigl( e^{-\be H_{\La} + |\La|\,G(\bar X_\La, \bar Y_\La)} \bigr) \Bigr|
  \leq   \norm W_{\La} \norm
\eeq and hence if one chooses $W_\La$ such that $\lim_{\La
\nearrow \bbZ^d} \norm W_{\La} \norm / \str \La| = 0$, then we
can replace $-\be H_\La$ by $-\be H_\La + W_\La$ in Theorem \ref{thm: main} and Theorem \ref{prop: mean field var principle}.    

 One also sees that it suffices to prove Theorem \ref{thm: main} and Theorem \ref{prop: mean field var principle} for the case that 
$\La = (a[-L,L])^d$ and $a \nearrow \infty$.  Convergence for $\La \nearrow \bbZ^d$ in the sense of Van Hove then follows  by the above log-trace inequality. 
\vspace{1mm}

\textbf{Finite-range restrictions.} 
It is obvious that our paper contains some restrictions that are not essential. 
In particular, by standard estimates (in particular, those used to prove the existence of the pressure, see e.g.\ \cite{simon}) one can relax the finite-range conditions on the interaction $\Phi$ to the condition that
\beq
\sum_{A \ni 0}   \frac{\norm \Phi_A \norm}{ \str A \str}  < \infty,
\eeq
and similarly for the local observables $X,Y$.
Moreover, it is not necessary that $G$ is a noncommutative polynomial. Starting from \eqref{eq: log-trace inequality}, one checks that it suffices that $G$ can be approximated in operator norm by noncommutative polynomials.

%

\section{Approximation by ergodic states} \label{sec: approx by ergodic}

In this section, we describe a construction that is the main ingredient of  our proofs, as well as of those in \cite{hiaimosonyiohno} and \cite{petzraggioverbeure}. This construction will be used in Sections \ref{sec: upper bound} and \ref{sec: proof of PRV}.

Let $V$ be a hypercube centered at the origin, i.e., $V= [-L, L]^d  $ for some $L>1$ and let 
\beq
 \partial V :=\{  i \in V \, \big\str \,  \exists   i' \in \bbZ^d \setminus V \, \textrm{such that}  \,  i, i'  \, \textrm{are nearest neighbours} \}
\eeq

Consider a state $\mu_V$  on $\caB_{V}$.

We aim to build an infinite-volume ergodic state out of $\mu_V$. First, we cover the infinite lattice $\bbZ^d$ with translates of the hypercube $V$ such that $\bbZ^d = \cup_{i \in (L\bbZ)^d} V+i$.
Then, we define the block product state
\beq
\tilde \mu := \mathop{\otimes}\limits_{ i \in (\bbZ L)^d}  \mu_V \circ \tau_i  \eeq
where $\mu_V \circ \tau_i$ is a state (translate of $\mu_V$) on  $\caB_{i+V}$.
We define also the 'translation-average' of $\tilde \mu$, 
\beq
\bar{\mu}:=   \frac{1}{\str V \str}\sum_{j \in V}   \tilde \mu \circ  \tau_j 
\eeq
We can now check the following properties:
\begin{itemize}
  \item  We have the exact equality of entropies
 \beq
s(\bar{\mu}) = s(\tilde{\mu}) =  \frac{1}{\str V \str} S(\mu_V)
 \eeq
 This follows  from the affinity of the entropy in infinite-volume.
A remark is in order: A priori, the infinite-volume entropy is defined for translation-invariant states, whereas $\tilde \mu$ is only periodic.  However, one easily sees that the entropy can still be defined, e.g.\ be viewing  $\tilde \mu$ as translation-invariant with a bigger one-site space, and correcting the definition by the 'volume' of this one-site space.
\item The state $\bar{\mu}$ is ergodic.  This follows, for example, from Proposition I.7.9 in \cite{simon}, as is also shown in \cite{hiaimosonyiohno} via an explicit calculation.  Note however that the $\bar{\mu}$ is in general not ergodic with respect to the  translations over the sublattice $(L \bbZ)^d$.  This phenomenon (though in a slightly different setting) is commented upon in  \cite{simon} (the end of Section III.5).
\item  The state $\bar{\mu} $ is a good approximation of $\mu_V$ for observables which are empirical averages,  provided $V$ is large.  Consider the local observable $X$ as in Section \ref{sec: hamiltonian and observables}. 
A translate $\tau_j X$ can lie inside a translate of $V$, i.e.\ $\supp \tau_j X \subset V+i$ for some $i \in \bbZ^d/V= (L\bbZ)^d$, or it can lie on the boundary between two translates of $V$.
The difference between  $\bar{\mu} (X)  = \bar{\mu} (\bar X)  $ and $ \mu_V(\bar X_V) $ clearly stems from those translates where $X$ lies on a boundary, and the fraction of such translates is bounded by
\beq
  \str \supp X\str  \times \frac{ \str \partial V \str  }{\str V \str}
\eeq
Hence 
\beq \label{eq: closeness of local A}
\left\str \bar{\mu} (\bar X)  - \mu_V(\bar X_V)  \right\str \leq \norm X \norm   \str \supp X\str  \times \frac{ \str \partial V \str  }{\str V \str}
\eeq
and also $ \left\str \bar{\mu} (X) -  \bar{\mu} (\bar X_{\La}) \right\str \leq  \norm X \norm   \str \supp X\str  \times \frac{ \str \partial \La \str  }{\str \La \str}$.

\end{itemize}

\section{The lower bound} \label{sec: lower bound}

In this section, we prove the following lower bound.  
\begin{lemma}
Recall $ Z^G_{\La}(\Phi)$ as defined in \eqref{def: ZG}. Then 
\beq
\liminf_{\La \nearrow \bbZ^d} \frac{1}{\str \La \str}\log  Z^G_{\La}(\Phi)  \geq \mathop{\sup}_{\om \in  \caS_{\mathrm{erg}} } ( (g(\om(X),\om( Y ))  +s (\om) - \om(E_\Phi) )
\eeq
where all symbols have the same meaning as in Section \ref{sec: results}.
\end{lemma}

\begin{proof}
Consider a  state $\om \in  \caS_{\mathrm{erg}}$.
We show that
\beq  \label{eq: proof lower bound one state}
\liminf_{\La \nearrow \bbZ^d} \frac{1}{\str \La \str}\log  Z^G_{\La}(\Phi) Ê\geq  g(\om(X),\om( Y ))  +s (\om) - \om(E_\Phi) 
\eeq
Consider, for each volume ${\La}$, the restriction $\om_{\La}:= \om\big\str_{\caB_{\La}}$.   By the  finite-volume variational principle 
\beq
 \frac{1}{\str {\La} \str}\log  Z^G_{{\La}}(\Phi) Ê\geq  \om_{\La} (G(\bar X_{\La},\bar Y_{\La} )  )  + \frac{1}{\str {\La} \str} S(\om_{\La}) -   \frac{1}{\str {\La} \str}  \om_{\La}(H_{\La})
\eeq
 The following convergence properties apply  (${\La} \nearrow \bbZ^d$ in the sense of Van Hove): 
\ben
\item{
\beq
\om_{\La} (G(\bar X_{\La},\bar Y_{\La} )  ) = \om(G(\bar X_{\La},\bar Y_{\La} )  )  \to g(\om(X),\om( Y )), \qquad {\La} \nearrow \bbZ^d   \label{eq: convergence of observable}
\eeq
  }
\item{
\beq
 \frac{1}{\str {\La} \str} S(\om_{\La})  \to s (\om), \qquad {\La} \nearrow \bbZ^d      \label{eq: convergence of entropy}
\eeq
  }
\item{
\beq
    \frac{1}{\str {\La} \str} \om(H_{\La})   \to  \om(E_{\Phi}), \qquad {\La} \nearrow \bbZ^d   \label{eq: convergence of energy}
\eeq
  }
\een
The relation \eqref{eq: convergence of energy} is obvious from the summability of $\Phi$, see Section \ref{sec: hamiltonian and observables}. The convergence  \eqref{eq: convergence of entropy} is the definition of the mean entropy $s$. Finally,  \eqref{eq: convergence of observable} follows from the ergodicity of  $\om$ as explained in Section \ref{sec: states}.

The relation \eqref{eq: proof lower bound one state} now follows immediately, since one can repeat the above construction for any ergodic state $\om$.
\end{proof}

\section{The upper bound} \label{sec: upper bound}
\subsection{Reduction to product states}

In this section, we outline how to approximate
\beq
\frac{1}{\str \La \str}\log  Z^G_{\La}(\Phi)
 \eeq
by a similar expression involving the partition function of a block-product state.
Fix a hypercube $V=[-L,L]^d$ and cover the lattice with its translates, as explained in Section \ref{sec: approx by ergodic}. From now on, $\La$ is chosen such that it is a 'multiple' of $V$, which is sufficient by the remark at the end of Section \ref{sec: results}.  Define the observables
\beq
H_{\La}^{V}\equiv  H_{\La}^{\Phi,V} , \qquad \bar X_{\La}^V, \qquad    \bar Y_{\La}^V
\eeq
by cutting all terms that connect any two translates of $V$, i.e.,
\beq
\bar X_{\La}^V :=    \frac{1}{\str \La \str} \sum_{\footnotesize{ \left.\begin{array}{c} j \in \La  \\  \exists i \in \bbZ^d/V:   \supp \tau_j A \subset V + i        \end{array}\right. }}    \tau_j X.
\eeq
and analogously for $H_{\La}^{V}$ and $\bar Y_{\La}^V$. One can say that these observables with superscript $V$ are 'one-block' observables with the blocks being translates of $V$.
One  easily derives that 
\beq
\norm  \bar X_{\La}^V- \bar X_{\La} \norm \leq      \norm X \norm   \str \supp X \str  \frac{ \str\partial V\str }{   \str V \str    }, \qquad   \norm H_{\La}^{V}-H_{\La} \norm \leq     r(\Phi)    \str \La \str \frac{ \str\partial V\str }{   \str V \str    }
\eeq
with the number $r (\Phi )$ as defined in Section \ref{sec: hamiltonian and observables}.

Using the log-trace inequality, we bound
\beq \label{eq: difference to product}
\frac{1}{\str \La \str} \log  \Tr_\La
  \left( e^{-H_\La+ |\La|\,G(\bar X_\La,\bar Y_\La )} \right)-  \frac{1}{\str \La \str} \log \Tr_\La \left( e^{-H_\La^V  + \str \La \str G(\bar X_{\La}^V, \bar Y_{\La}^V)}    \right)
\eeq
as follows
\baq
 \textrm{ \eqref{eq: difference to product}  } &\leq &   \frac{1}{\str \La \str} \norm H_\La- H_\La^V \norm  +    \norm G(\bar X_\La,\bar Y_\La )-    G(\bar X_{\La}^V, \bar Y_{\La}^V) \norm  \nonumber  \\
&\leq &     \big(  r (\Phi )+   C_g  (\norm X \norm   \str \supp X \str+ \norm Y \norm   \str \supp Y \str) \big)  \frac{ \str\partial V\str }{   \str V \str    }    \nonumber    \eaq
where  $C_g$  is constant depending on the function $G$.
The second term of \eqref{eq: difference to product} is clearly the pressure of a product state with mean field interaction. We will find an upper bound for this pressure by slightly extending the treatment of Petz et al.\ in \cite{petzraggioverbeure}. We prove an  'extended PRV' -lemma, Lemma \ref{lem: PRV} in the next section.
\subsection{ The extended Petz-Raggio-Verbeure upper bound}

In this section, we outline the bound from above on the quantity
\beq
\frac{1}{\str \La \str} \log \Tr_\La \left( e^{-H_\La^V  + \str \La \str G(\bar X_{\La}^V, \bar Y_{\La}^V)}    \right)
\eeq
that appeared in \eqref{eq: difference to product}.

To do this, let us make the setting slightly more abstract.
Consider again the  lattice $\bbZ^d$ with the one-site Hilbert space $\caG$, which should be thought of as
\beq
\caG:= \otimes_{V} \caH
\eeq
In words, the  sites  of the new lattice are actually blocks of the old lattice. In particular, the trace $\Tr_{\La}$ has a slightly different meaning as before since the one-site Hilbert space has changed, i.e.\ the trace is on the Hilbert space $\otimes_{\La}\caG$ instead of $\otimes_{\La}\caH$.
Let $D, A,B$ be  one-site observable on the new lattice, i.e.\ $D, A,B$ are Hermitian operators on  $\caG$.
 The
extended PRV (Petz-Raggio-Verbeure)  states that
\begin{lemma}[Extended PRV]\label{lem: PRV}
Let all symbols have the same meaning as in Sections \ref{sec: hamiltonian and observables}-\ref{sec: noncommutative functions}-\ref{sec: states}, except that the one-site Hilbert space is changed from $\caH$ to $\caG$. Then
\beq \label{eq: extended prv}
\limsup_{\La \nearrow \bbZ^d}  \frac{1}{\str \La \str} \Tr_\La \left( e^{-D_\La  + \str \La \str G(\bar A_{\La}, \bar B_{\La}) }    \right)   \leq \mathop{\sup}\limits_{\om \in \caS_{\textrm{sym}} } \left(  \om(G(\bar A,\bar B)) +s(\om) -\om(D) \right)
\eeq
In particular $\om(G(\bar A, \bar B))$ defined as \eqref{eq: limit k and l} exists.
\end{lemma}
To appreciate the similarity between  \eqref{eq: extended prv} and \eqref{eq: prop mean field variational principle}, one should realize that $D$ is a local energy operator, as $E_{\Psi}$ in \eqref{eq: prop mean field variational principle}.
The proof of this lemma in the case that $A=B$ is in the original paper \cite{petzraggioverbeure}. The proof for the more general case is presented in Section \ref{sec: proof of PRV}. 
Of course, one can prove that the  RHS of \eqref{eq: extended prv} is also a lower bound: it suffices to copy Section \ref{sec: lower bound}.

By the St{\o}rmer theorem, see \eqref{eq: stormer}, each symmetric state $\om$ on $\frU$ is a convex combination of product states, and since all terms on the RHS of \eqref{eq: extended prv} are affine functions of $\om$, it follows that the $\sup$ can be restricted to  product states. Since, moreover, all product states are ergodic, we can replace $\om(G(\bar A,\bar B))$ by $g(\om(A), \om(A))$. Hence, Lemma \ref{lem: PRV} implies that 
\beq \label{eq: extended prv cosequence}
\limsup_{\La \nearrow \bbZ^d}  \frac{1}{\str \La \str} \Tr_\La \left( e^{-D_\La  + \str \La \str G(\bar A_{\La}, \bar B_{\La}) }    \right)   \leq \mathop{\sup}\limits_{\om \, \textrm{prod.}}  \left( g(\om(A), \om(A)) + s(\om) -\om(D) \right)
\eeq

\subsubsection{From the extended PRV to the upper bound}
Next, we use \eqref{eq: extended prv cosequence} to formulate an upper bound on the quantity
\beq \label{eq: to be upper bounded}
\frac{1}{\str \La \str} \Tr_\La \left( e^{-H_\La^V  + \str \La \str G(\bar X_{\La}^V, \bar Y_{\La}^V)}    \right)
\eeq
for $\La$ a multiple of $V$. This means that we have to recall that the lattice sites in \eqref{eq: extended prv cosequence} are in fact blocks.
 We write $\La^{*}:= \La /V$
and choose
\baq
D &:=&   H_V   \nonumber \\
A &:=&   \bar X_V   \nonumber  \\
B  &:=&    \bar   Y_V.    \nonumber
\eaq
Then, by the extended PRV,
\baq
\textrm{\eqref{eq: to be upper bounded}} &\leq&  \sup_{\om \, \textrm{prod. on} \,  \caB(\La^{*}) }       \left(  g(\om(\bar A),\om(\bar  B)) +   \frac{1 }{ \str V \str}    s^*(\om) -  \frac{1 }{ \str V \str } \om(D)   \right)   \nonumber \\[2mm]
  &=&  \sup_{\om_V\,  \textrm{on}\,  \caB_V}       \left(   g (  \om_V(\bar X_V) ,   \om_V(\bar   Y_V ))  +    \frac{1 }{ \str V \str }    S(\om_V) -  \frac{1 }{ \str V \str } \om_V( H_V)   \right)  \nonumber
\eaq
where $s^*$ indicates that this is the entropy density on the block lattice $\La^*$, hence it should be divided by $\str V\str$ to obtain the density on $\La$.  Now, let $\tilde\om$ be the infinite-volume state obtained by taking a block-product over states $\om_V$ and let $\bar{\om}$ be its 'translation-average', as in Section \ref{sec: approx by ergodic}. By the conclusions of Section \ref{sec: approx by ergodic}, it follows that $\bar{\om}$ is ergodic and $s(\bar{\om})= S(\om_V)  $.
 Also, we see that
\baq
\str \om_V(\bar X_V)  - \bar{\om} (X) \str   & \leq & \norm X \norm  \str \supp X \str  \frac{\str \partial V \str}{\str V \str} \nonumber \\
 \str \om_V(H_V)  -  \bar{\om}(E_{\Phi})   \str   & \leq & r( \Phi)   \frac{\str \partial V \str}{\str V \str}  \nonumber
\eaq
and analogously for  $ \bar Y_V$.
Consequently, we obtain
\beq
\textrm{\eqref{eq: to be upper bounded}}  \leq\mathop{\sup}_{\om \in  \caS_{\mathrm{erg}} } \left(   g (  \om(\bar X) ,   \om(\bar   Y )) +   s(\om) - \om(E_{ \Phi})   \right)  + O(  \frac{\str \partial V \str}{\str V \str}), \qquad V \nearrow \bbZ^d   \nonumber
\eeq
which proves the upper bound  for Theorem \ref{prop: mean field var principle}, since the $O(  \frac{\str \partial V \str}{\str V \str})$-term can be made arbitrarily small by increasing $V$. 

\section{Proof of Lemma \ref{lem: PRV}} \label{sec: proof of PRV}

Let the state $\mu_{\La}$ on $\caB_{\La}$ be given by
\beq
\mu_{\La}(\cdot) = \frac{1}{Z^{G}_{\La}(D)} \Tr_{\La} \left( e^{-D_{\La}  + \str {\La} \str G(\bar A_{{\La}}, \bar B_{{\La}}) }   \cdot \right) \nonumber
\eeq
with
\beq
Z^{G}_{\La}(D) :=  \Tr_{\La} \left( e^{-D_{\La}  + \str {\La} \str G(\bar A_{{\La}}, \bar B_{{\La}}) }   \right) \nonumber.
\eeq
Naturally, $\mu_{\La}$ is the finite-volume Gibbs state that saturates the variational principle, i.e.\
\baq \label{eq: prv finite volume gibbs}
\frac{1}{\str {\La} \str}  \log Z^{G}_{\La}(D) &  =&  \sup_{ \om_{\La}\, \textrm{on}\, \caB_{\La} } \left(\om_{\La} ( G(\bar A_{\La},\bar B_{\La}) ) + \frac{1}{\str {\La} \str} S(\om_{\La}) -\om_{\La}(D) \right)  \nonumber \\
 &=&  \mu_{\La} ( G(\bar A_{\La},\bar B_{\La}) ) +\frac{1}{\str {\La} \str} S(\mu_{\La}) - \mu_{\La}(D)
\eaq

Our strategy is  to attain the 'entropy' and 'energy' of the state $\mu_{\La}$ via ergodic states. For definiteness, we assume that $G$ is of the form
\beq
G(\bar A_{{\La}}, \bar B_{{\La}}):=[\bar A_{{\La}}]^{k} [\bar B_{{\La}}]^l \qquad \textrm{ for some integers} \, \, k,l, \nonumber
\eeq
(which, strictly speaking, is not allowed since $G(\bar A_{{\La}}, \bar B_{{\La}})$ has to be a self-adjoint operator, but this does not matter for the argument in this section).
The general case  follows by the same argument.

Since $\mu_{\La}$ is a symmetric state on $\caB_{\La}$ (hence translation invariant in the restricted sense of Section \ref{sec: approx by ergodic}), we can apply the construction in Section \ref{sec: approx by ergodic}, which yields us infinite-volume states
$\tilde \mu$ and $\bar{\mu}$.  Since we will repeat the construction for different ${\La}$, we indicate the ${\La}$-dependence in  $\tilde \mu^{\{{\La}\}}$ and $\bar{\mu}^{\{{\La}\}}$, but remembering that these are states on the infinite lattice.
They satisfy
\beq
s(\bar{\mu}^{\{{\La}\}})  =  \frac{1}{\str {\La} \str}S(\mu_{\La})
\eeq
We have also established in Section \ref{sec: approx by ergodic} that $\bar{\mu}^{\{{\La}\}}$ is ergodic and that the states $\bar{\mu}^{\{{\La}\}}$ and $\tilde{\mu}^{\{{\La}\}}$ approximate $\mu_{\La}$ for observables which are empirical averages.  However, we cannot conclude yet that they have comparable values for $G(\bar A, \bar B)$, except in the case where $G$ is linear. Essentially, such a comparison is achieved next by using the fact that $\mu_{\La}$ is symmetric.

Choose a sequence of volumes ${\La}_n$ such that along that sequence the RHS of \eqref{eq: prv finite volume gibbs} converges. Be assume that $\bar\mu^{{\La}_n}$  has a weak-limit, as $n \nearrow \infty$, which can always be achieved (by the weak compactness) by restricting to  a subsequence. of ${\La}_n$.  We call this limit $\mu$. By construction, it is a symmetric state. \\

\noindent \textbf{Energy estimate:} 
Since $\bar\mu^{{\La}_n} \rightarrow \mu$, weakly, and  $\bar\mu^{{\La}_n}(D)=\mu_{{\La}_n}(D)$, we have 
\beq  \label{eq: convergence of energies}
\mu_{{\La}_n} (D) \rightarrow  \mu (D)
\eeq \\

\noindent \textbf{G-estimates:}
Using the symmetry of the state $\mu_{\La}$, we estimate
\beq  \label{eq: G under finite symmetric}
\mu_{\La} (G(\bar A_{\La},\bar B_{\La})) =  \mu_{\La}(  \otimes^k A \otimes^l B   ) +O(\frac{\str k+l\str}{\str {\La} \str}) \max{(\norm A \norm, \norm B \norm)}^{k+l}
\eeq
where the tensor products
\beq
\otimes^k A \otimes^l B  :=  \underbrace{A \otimes  \ldots \otimes A}\limits_{k \, \textrm{copies}} \otimes   \underbrace{B \otimes \ldots \otimes  B}\limits_{l \, \textrm{copies}}
 \eeq
 denote that all one-site operators are placed on different sites. Since $\mu_{\La}$ is symmetric, we need not specify on \emph{which} sites.  The error term comes from those terms in the expansion of the monomial where two one-site operators hit the same site.
 Since $\mu$ is symmetric, we obtain analogously that
\beq   \label{eq: G under infinite symmetric}
\mu (G(\bar A,\bar B)) =    \mu( \otimes^k A \otimes^l B  )
\eeq
In particular, the LHS is well-defined.
Hence,  by combining \eqref{eq: G under finite symmetric} and \eqref{eq: G under infinite symmetric}, we obtain
\beq     \label{eq: convergence of G}
\mu_{{\La}_{n} } (G(\bar A_{{\La}_{n} },\bar B_{{\La}_{n} }))  \rightarrow   \mu (G(\bar A,\bar B)).
\eeq
For a more general noncommutative polynomial $G$ as defined in Section \ref{sec: noncommutative functions}  (not necessarily a monomial),  the convergence \eqref{eq: convergence of G} follows easily since $G(\bar A_{{\La}_{n} },\bar B_{{\La}_{n} })$ can be approximated in operator norm by polynomials.  \\

\noindent \textbf{Entropy estimates:}
As established in Section \ref{sec: approx by ergodic}, we have
\beq
 \frac{1}{\str {\La} \str} S(\mu_{{\La}}) =    s(\bar{\mu}^{\{{\La}\}}),\qquad \textrm{for all} \, \, {\La}
\eeq
By the lower semi-continuity of the infinite-volume entropy and the convergence $\bar{\mu}^{{\La}_n} \rightarrow \mu$, we get that
\beq
\mathop{\limsup}\limits_{n\nearrow \infty}   s(\bar{\mu}^{\{{\La}_n\}})   \leq  s(\mu )
\eeq
Hence
\beq \label{eq: convergence of entropies}
\mathop{\lim}\limits_{n\nearrow \infty}   \frac{1}{\str {\La}_n \str}  S(\mu_{{\La}_{n}})   \leq  s(\mu)
\eeq

By combining the convergence results  (\ref{eq: convergence of energies}, \ref{eq: convergence of G}, \ref{eq: convergence of entropies}), we have proven that there is a symmetric state $\mu$ such that the RHS of \eqref{eq: extended prv} with $\om \equiv \mu$ is larger than a given limit point of the RHS of \eqref{eq: prv finite volume gibbs}. Since the construction can be repeated for any limit point, this concludes the proof of Lemma \ref{lem: PRV}.

\vspace{1cm}
\subsection*{Acknowledgment} The authors thank M.~Fannes, M.~Mosonyi,  Y.~Ogata, D.~Petz and A.~Verbeure for fruitful discussions.  K.~N.~ is also grateful to the Instituut voor Theoretische Fysica, K.~U.~Leuven, and to Budapest University of Technology and Economics for kind hospitality, and acknowledges the support from the Grant Agency of the Czech Republic (Grant no.~202/07/J051).
W.~D.~R.~ acknowledges the support of the FWO-Flanders. L.R.-B. acknowledges the support of the NSF  
(DMS-0605058)
\vspace{1cm}

\appendix

\section{Proof of Lemma \ref{lem: concentration}} \label{app: concentration}

To prove Lemma \ref{lem: concentration}, it is convenient to introduce an extended framework:  Let  $\pi_\om$ be the cyclic GNS-representation associated to the state $\om$,  $\frH_{\om}$ the associated Hilbert space and $\psi \in \frH_\om$  the representant of the state $\om$, i.e.\
\beq
\om (A) = \langle \psi, \pi_{\om}(A) \psi \rangle_{\frH_{\om}}, \qquad  A \in \frU
\eeq
The set $\pi_{\om}(\frU)$ is a subalgebra of $\caB(\frH_{\om})$. 
Let $U_{j}, \in \bbZ^d$ be the unitary  representation of the translation group induced on $\pi_{\om}(\frU)$, i.e.\ 
\beq
U_j \pi_{\om}(A) U_{j}^* =    \pi_{\om}( \tau_j A).
\eeq
Ergodicity of $\om$  implies (see e.g.\ the proof of  Theorem III.1.8  in  \cite{simon})   that
\beq \label{eq: von neumann ergodic}
  \frac{1}{\str\La\str}\sum_{j \in \La}  U_j        \quad \mathop{\longrightarrow}\limits_{\La \nearrow \bbZ^d}^{\mathrm{strongly}}  \qquad  P_{\psi}
\eeq
where $P_\psi$ is the one-dimensional orthogonal projector associated to the vector $\psi$, and $\La \nearrow \bbZ^d$ in the sense of Van Hove.
Using \eqref{eq: von neumann ergodic} and the translation-invariance $U_j \psi=\psi$, one calculates
\baq
  \pi\left( \bar X_{\La} \right)   \pi\left( \bar Y_{\La} \right) \psi    &  = &  \frac{1}{\str\La\str^2}      \sum_{j, j' \in \La}  U_j   \pi(X) U_{j'-j}  \pi(Y)  U_{-j'} \psi   \nonumber  \\  
   &   \quad \mathop{\longrightarrow}\limits_{\La \nearrow \bbZ^d}    \quad &  P_{\psi} \pi(X)  P_{\psi}  \pi(Y)  \psi  =  \om(X) \om(Y)  \psi        \nonumber
\eaq
for local observables $X, Y \in \frU$. Taking the scalar product with $\psi$, we conclude that $\om (\bar X_{\La}  \bar Y_{\La}  )  \rightarrow \om (X) \om(Y ) $. The same argument works for all  polynomials in $\bar X_{\La},  \bar Y_{\La} $, thus proving Lemma \ref{lem: concentration}.
Finally, we remark that one can also construct the operators $\bar X, \bar Y$ as weak limits of $\bar X_{\La},  \bar Y_{\La} $, as $\La \nearrow \bbZ^d$ (these weak limits are simply multiples of identity:  $ \om(X) 1, \om(Y) 1$). This is however not necessary for our results.

\bibliographystyle{plain}
\bibliography{mylibrary08}

\end{document}